%% file: main.tex
\documentclass[final,finalrunningheads,a4paper]{llncs}
\input{styles.tex}

\input{macro.tex}

\begin{document}

\title{\mbox{\hspace{-14pt}Certified algorithms for numerical semigroups in Rocq}}

\author{Massimo Bartoletti\inst{1} \and
Stefano Bonzio\inst{1} \and
Marco Ferrara\inst{1}}
\institute{Universit\`a degli Studi di Cagliari, Italy}

\maketitle

\input{abstract}


\input{intro.tex}
\input{def.tex}
\input{generators.tex}

\input{conclusions.tex}

\bibliographystyle{splncs04}
\bibliography{main}

\end{document}

%% file: styles.tex
\usepackage[T1]{fontenc}	
\usepackage[utf8]{inputenc}	
\usepackage[english]{babel}	

\usepackage{amsmath}
\usepackage{amssymb}
\usepackage{braket}
\usepackage{mathtools}
\usepackage{listings}
\usepackage[noEnd=true]{algpseudocodex}
\usepackage{algorithm}

\usepackage{booktabs}
\usepackage{caption}

\usepackage{xspace}
\usepackage{graphicx}

\usepackage{thmtools}

\usepackage{xifthen}  
\usepackage{ifthen}

\usepackage[colorlinks]{hyperref}
\usepackage{cleveref}

\hypersetup{
  breaklinks   = true,
  colorlinks   = true, 
  urlcolor     = blue, 
  linkcolor    = blue, 
  citecolor    = red   
}
\hypersetup{final} 

\usepackage[final,nomargin,inline,index]{fixme} 
\fxusetheme{color}
\FXRegisterAuthor{bart}{anbart}{\color{magenta} {\underline{bart}}}

\FXRegisterAuthor{ste}{anste}{\color{red} {\underline{ste}}}

\FXRegisterAuthor{marco}{anmarco}{\color{blue} {\underline{marco}}}

\usepackage{hyperref}

%% file: macro.tex
\newcommand{\ifempty}[3]{%
  \ifthenelse{\isempty{#1}}{#2}{#3}%
}

    \newcommand{\citeCoqRoot}{https://github.com/marcoferr99/NumSgps/blob/cicm}
\newcommand{\citeCoqLink}[1]{\href{\citeCoqRoot/#1}{\raisebox{-3pt}{\includegraphics[height=2.75ex]{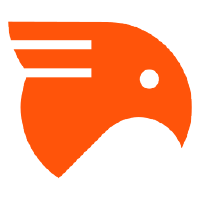}}}}
\newcommand{\citeCoq}[2][]{\textup{\textbf{(}}\ifempty{#2}{\!}{\citeCoqLink{#2}}\ifempty{#1}{}{\;\textup{\lstinline[language=C]{#1}}}\textup{\textbf{)}}}

\newcommand{\Eg}{E.g.\@\xspace}
\newcommand{\eg}{e.g.\@\xspace}
\newcommand{\ie}{i.e.\@\xspace}
\newcommand{\wrt}{w.r.t.\@\xspace}
\newcommand{\qedex}{\ensuremath{\diamond}}
\newcommand{\keyterm}[1]{\textbf{\textit{#1}}}

\newcommand{\listEmpty}{\mbox{\ensuremath{[\,]}}}
\newcommand{\listCons}[2]{\mbox{\ensuremath{{#1}::{#2}}}}
\newcommand{\listConcat}{\,}
\newcommand{\app}[2]{{#1} \listConcat {#2}}

\newcommand{\natlists}{\mathbb{L}}
\newcommand{\ltlists}{\natlists_<}
\newcommand{\gelists}[1]{\natlists_\ge^#1}

\newcommand{\listL}[1][]{\vec{s}_{#1}}
\newcommand{\listLi}[1][]{\vec{s'}_{#1}}
\newcommand{\listT}[1][]{\vec{t}_{#1}}
\newcommand{\listTi}[1][]{\vec{t'}_{#1}}
\newcommand{\listG}[1][]{\vec{g}_{#1}} 
\newcommand{\listA}[1][]{\vec{a}_{#1}} 
\newcommand{\gen}[1][]{\listA[{#1}]}
\newcommand{\gaps}[1]{\vec{g}_{#1}}

\newcommand{\N}{\mathbb{N}}
\newcommand{\Z}{\mathbb{Z}}

\newcommand{\irule}[2]{\dfrac{#1}{#2}}
\newcommand{\setenum}[1]{\{#1\}}
\newcommand{\setcomp}[2]{\left\{{#1} \,\middle|\, {#2}\right\}}

\newcommand{\ran}[1]{\operatorname{Im}\ifempty{#1}{}{({#1})}}
\newcommand{\eqdef}{\triangleq}

\newcommand{\findgap}[2][]{f_{#1}\ifempty{#2}{}{({#2})}}
\newcommand{\findmod}[3]{\mathrm{fm}_{#1}(#2, #3)}

\newcommand{\rep}[2]{[{#1}]^{#2}}
\renewcommand{\next}[2][]{\mathsf{next}_{#1}\ifempty{#2}{}{({#2})}}
\newcommand{\nh}[2]{\mathsf{genL}_{#1}\ifempty{#2}{}{({#2})}}
\newcommand{\nhs}[2]{\mathsf{genM}_{#1}\ifempty{#2}{}{({#2})}}

\newcommand{\apery}[2]{\mathsf{Ap}(#1, #2)}
\newcommand{\aperytwo}[2]{\mathsf{Ap2}(#1, #2)}
\newcommand{\length}[1]{|{#1}|}
\newcommand{\genlength}{m}
\newcommand{\ml}[2]{|{#2}|_{\geq {#1}}}
\newcommand{\mls}[2]{|{#2}|_{\ge 0, \dots, #1}}
\newcommand{\all}[3]{A^{#1}_{#2}(#3)}
\newcommand{\fsucc}[1]{S_{#1}}

\spnewtheorem*{notation}{Notation}{\itshape}{\rmfamily}

%% file: abstract.tex
\begin{abstract}
A numerical semigroup is a co-finite submonoid of the monoid of non-negative integers under addition.
Many properties of numerical semigroups rely on some fundamental invariants, such as, among others, the set of gaps (and its cardinality), the Apéry set or the Frobenius number.
Algorithms for calculating invariants are currently based on computational tools, such as GAP, which lack proofs (either formal or informal) of their correctness.
In this paper we introduce a Rocq formalization of numerical semigroups. Given the semigroup generators, we provide certified algorithms for computing some of the fundamental invariants:
the set of gaps, of small elements, the Apéry set, the multiplicity, the conductor and the Frobenius number.  
To the best of our knowledge this is the first formalization of numerical semigroups in any proof assistant.
\end{abstract}

%% file: intro.tex
\section{Introduction}

A numerical semigroup is a co-finite additive submonoid of the monoid of the natural numbers $\mathbb{N}$.
The main algebraic reason for the investigation of this class of structures traces back to the so-called ``Frobenius problem'', namely finding the largest $b\in\mathbb{N}$ such that a Diophantine equation $a_{1}x_{1} + \dots + a_{n}x_{n} = b$, with $a_{1},\dots, a_{n}$ coprime natural numbers, has no positive integer solution~\cite{FrobProblem}.
Nowadays numerical semigroups play a key role in commutative algebra and algebraic geometry.
For example, they provide much information about rings of coordinates originated from algebras of polynomials over a field \cite{Barucci}, and (free) numerical semigroups are associated to singularities of certain planar curves \cite{Garcia1982,Campillo}.  
More recently, numerical semigroups have also been fruitfully applied to cryptography and coding theory \cite{PedroCodes}. 

Given a numerical semigroup $M$, the (finite) set $\mathbb{N}\setminus M$ is called the set of \emph{gaps} of $M$; every numerical semigroup is univocally determined by the set of its gaps. Moreover, every numerical semigroup is finitely generated \cite[Proposition 3]{numsgps}. One of the advantages of computer-aided approaches to algebraic structures, in general, relies on the possibility of producing (counter)examples on whose basis conjectures can be formulated, checked or disproved: since numerical semigroups are finitely generated, examples can be obtained by algorithms that produce a numerical semigroup starting from a (finite) set of generators. 
One natural way to do that --- which we implement in Rocq --- is by calculating the set of gaps starting from the generators. The inverse problem, namely to compute a (finite) set of generators starting from a given numerical semigroup, turns out to be interesting as well. To this end, we will focus on computing the Apéry set of a numerical semigroup, which is an invariant carrying more information than the set of minimal generators~\cite{numsgps}.

The study of numerical semigroups, as well as the current development of the related theory, strongly relies on the aid of computational tools: the most popular of which, in the algebraic community, is the GAP \cite{GAP} package $\textsf{numericalsgps}$ \cite{GAPPackageNumSem}. 
Algorithmic methods are widely adopted in the theory of numerical semigroups (see, \eg, \cite{Counting,Oversemigroups,Fundamentalgaps}): for example, the work \cite{algonumsgps} studies various algorithms for generalized numerical semigroups. 
However, the existing literature lacks, on the one side, formal proofs of the correctness of the developed algorithms 
and, on the other, the implementation of certified algorithms via the use of proof assistants. The main motivation of our work is to fill this gap by introducing the first (to the best of our knowledge) formalization of numerical semigroups in a proof assistant. 
More specifically, we develop certified algorithms in Rocq that compute: \begin{enumerate}

\item the multiplicity of a numerical semigroup (\Cref{sec:miltiplicity});

\item the Apéry set of a numerical semigroup (\Cref{sec: Apery});

\item a numerical semigroup from a (finite) set of generators (\Cref{sec:gen-to-small-elements}). 

\end{enumerate}

Besides mechanically proving the correctness for these algorithms in Rocq, we provide in the paper informal sketches of our proof strategies.
Throughout the paper we provide \citeCoq[links]{/} to the Rocq sources of our formalization. 


%% file: def.tex
\section{Computing numerical semigroups invariants}

The standard definition of numerical semigroup given before, albeit elegant, is not practical from a computational perspective, since it does not provide a finite representation that can be fed as input to an algorithm.
To define a numerical semigroup $M$ in Rocq, we then rely on the finite set $ \N \setminus M$ of \emph{gaps} of~$M$, 
%
which we represent as the list of its elements and refer to as $\gaps{M}$;  
we additionally require this list to be ordered and duplicate-free, so to ensure that each numerical semigroup has a unique representation. 

\begin{notation}
We introduce some notation that will be used throughout the paper.
We denote the empty list by $\listEmpty$, and the list constructed by appending a head element $h$ to a tail $\listT$ as $\listCons{h}{\listT}$ (note that $h \in \N$, while $\listT$ is a list of natural numbers).
\Eg, the expression $\listCons{0}{(\listCons{1}{(\listCons{2}{\listEmpty})})}$ stands for the list of the first three natural numbers.
For brevity, we will often write such list as $[0;1;2]$.
As usual, the concatenation of two lists is denoted by juxtaposition:
for example, $\app{[0;1;2]}{[5;4;8]} = [0;1;2;5;4;8]$.
Given a list $\listL$ and $k,i \in \N$,
we write $\rep{k}{i}$ for the list made of the element $k$ repeated $i$ times,
$\listL[i]$ for the $i$-th element of $\listL$ (starting from zero),
and $\ml{i}{\listL}$ for the number of elements in $\listL$ that are greater than or equal to $i$
(in particular, $\ml{0}{\listL}$ is the length of $\listL$, abbreviated as $\length{\listL}$).
We denote by $\natlists$ the set of lists of natural numbers,
by $\ltlists$ the subset of $\natlists$ containing only the lists that are duplicate-free and sorted in ascending order, and 
by $\gelists{m}$ the subset containing only the lists that are sorted in descending order and whose elements are less than or equal to the natural number $m$.
Note that $\gaps{M} \in \ltlists$.
\end{notation}

The standard literature on numerical semigroups \cite{numsgps,numsgps2} is obviously based on classical logic, thus the excluded middle law for the membership relation tells us 
that in a numerical semigroup $M$, for all $x \in \N$ either $x \in M$ or \mbox{$x \notin M$}. 
Even though excluded middle could be integrated in Rocq by explicitly assuming the classical axiom, this is not really necessary: it suffices to assume decidability for numerical semigroup membership only.
Actually, the definition of numerical semigroups via the list of gaps already includes this hypothesis, since list membership is decidable.

We now recap some key invariants associated to a numerical semigroup $M$:
\begin{itemize}

\item \label{keyterm:multiplicity}
\keyterm{multiplicity}: the minimum non-zero element of $M$;  

\item \label{keyterm:conductor}
\keyterm{Frobenius number}: the maximum element of $\Z \setminus M$; 

\item \keyterm{conductor}: the successor of the Frobenius number;

\item \label{keyterm:small-elements}
\keyterm{small elements}: the elements of $M$ less than or equal to the conductor;

\item \label{keyterm:apery} 
\keyterm{Apéry set}: given an $n \in \N$, they are the elements of $M$ such that $x-n \not\in M$; 


\end{itemize}

We define algorithms for computing all of these invariants, given the gaps list $\gaps{M}$ representing a numerical semigroup $M$.
Some of these algorithms are relatively simple and are briefly explained here, while the others are examined in greater depth in the rest of the paper.
We implement the conductor as follows: if the list of gaps $\gaps{M}$ is not empty then the conductor is the maximum element of $\gaps{M}$ plus one; otherwise, the conductor is $0$.
In our Rocq implementation we preferred not to work with integers but only with natural numbers; for this reason, we do not explicitly provide a function calculating the Frobenius number.
Note, however, that all the relevant results relying on the Frobenius (see \eg \cite{numsgps2}) can also be expressed via the conductor.
We compute the small elements of $M$ as the list obtained by only keeping the elements that belong to $M$ in the sequence of natural numbers between $0$ and the conductor.

\label{def:generators}
We say that a numerical semigroup $M$ is \keyterm{generated} by a set $A \subseteq \N$ if every element of $M$ can be written as a linear combination of elements of $A$, \ie for all $d \in M$ there exist $a_{1},\dots a_{n}\in A$ and $\lambda_{1},\dots, \lambda_{n}\in \N$ such that $d = \sum_{i=1}^{n}\lambda_{i}a_i$. 
Every numerical semigroup is generated by a \emph{finite} set \cite[Proposition 3]{numsgps}.

\subsection{Computing the multiplicity}
\label{sec:miltiplicity}

We compute the multiplicity by searching for the smallest natural number, starting from $1$, that does \emph{not} belong to the gaps list. 
To this purpose, we introduce a function $\findgap[1]{}$ that, taken as input a list $\listL$, gives as output the smallest integer greater than the last consecutive integers in $\listL$.
For example, $\findgap[1]{[1;2;5]} = 3$.
Technically, it is convenient to parameterise $\findgap[x]{}$ over an arbitrary $x \in \N$, so that $\findgap[x]{\listL}$ outputs $x$ for any list whose first element is different from $x$, and $n+1$ for any list of the form $[x; x+1;\dots; x+ n, y]$ with $y \neq x+(n+1)$. 

\begin{definition}
    \label{def:find_gap}
    \citeCoq[find_gap]{list_nat.v\#L33}
    For all $x \in \N$, let $\findgap[x]{} \colon \natlists \rightarrow \N$ be defined recursively as follows:
    \[
    \findgap[x]{\listEmpty} = x
    \qquad
    \findgap[x]{\listCons{x}{\listT}} = \findgap[x+1]{\listT}
    \qquad
    \findgap[x]{\listCons{h}{\listT}} = x \quad\text{if $h \neq x$}.
    \]
\end{definition}
For example, we have that
$\findgap[1]{[1;2;3;7]} = \findgap[2]{[2;3;7]} = \findgap[3]{[3;7]} = \findgap[4]{[7]} = 4$.

\Cref{th:multiplicity} states that $\findgap[1]{}$ gives the multiplicity of $M$.
For example, consider the numerical semigroup $M$ whose list of gaps is $\gaps{M} = [1;2;3;5;6;9;13]$.
The multiplicity of $M$ is $\findgap[1]{[1;2;3;5;6;9;13]} = 4$.

\begin{theorem}
    \label{th:multiplicity}
    \citeCoq[multiplicity_min]{def.v\#L59}
    Let $M$ be a numerical semigroup represented by the gaps list $\gaps{M}$.
    Then, $\findgap[1]{\gaps{M}}$ is the multiplicity of $M$.
\end{theorem}

\begin{proof}
We have to show that $\findgap[1]{\gaps{M}}$ is the minimum non-zero element of $M$.
We start by proving that $\findgap[1]{\gaps{M}} \neq 0$.
This is a special case of the lemma:
\begin{equation}
    \label{eq:find_gap_le}
    \forall x \in \N, \;\;
    \forall \listL \in \natlists, \;\;
    x \le \findgap[x]{\listL}
    \tag*{\citeCoq[find_gap_le]{list_nat.v\#L40}}
\end{equation}
We prove this by induction on the structure of $\listL$.
For the base case $\listL = \listEmpty$, we have that $\findgap[x]{\listL} = x \ge x$.
For the inductive case, let $\listL = \listCons{h}{\listT}$.
There are two subcases.
If $h \neq x$, then $\findgap[x]{\listCons{h}{\listT}} = x \ge x$.
Otherwise, if $h = x$, then $\findgap[x]{\listCons{h}{\listT}} = \findgap[{x+1}]{\listT}$ and by inductive hypothesis we have
\(
x \le x+1 \le \findgap[{x+1}]{\listT}
\).
\hfill\qedex

\medskip\noindent
By applying \ref{eq:find_gap_le}
for $\listL = \gaps{M}$ and $x=1$, we obtain $\findgap[1]{\gaps{M}} \neq 0$.

\medskip
We now prove that $\findgap[1]{\gaps{M}} \in M$, 
\ie $\findgap[1]{\gaps{M}} \notin \gaps{M}$.
This is a special case of a more general auxiliary result:
\begin{equation}
    \label{eq:find_gap_notin}
    \forall x \in \N, \;\;
    \forall \listL \in \ltlists, \;\;
    \ml{x}{\listL} = \length{\listL}
    \; \rightarrow \;
    f_x(\listL) \notin \listL
    \tag*{\citeCoq[find_gap_notin]{list_nat.v\#L59}}
\end{equation}

\noindent
We prove this by induction on the structure of $\listL$.
For the base case $\listL = \listEmpty$, the thesis holds trivially.
For the inductive case $\listL = \listCons{h}{\listT}$, there are two subcases:
\begin{itemize}

\item $h = x$. 
Then, $\findgap[x]{\listCons{h}{\listT}} = \findgap[{x+1}]{\listT}$.
Suppose by contradiction that $\findgap[{x+1}]{\listT} \in \listCons{h}{\listT}$.
If $f_{x+1}(\listT) = h$ then by \ref{eq:find_gap_le}
we have $x+1 \le f_{x+1}(\listT) = h = x$ --- contradiction.
Thus, we must have $\findgap[{x+1}]{\listT} \in \listT$.
Note that every element of $\listT$ is greater than or equal to $x+1$, since $\listCons{h}{\listT}$ is sorted and duplicate-free, and $h = x$. 
Therefore, by the inductive hypothesis it follows that $\findgap[x+1]{\listT} \not\in \listT$ --- contradiction. 
    
\item  $h \neq x$. Then, $\findgap[x]{\listCons{h}{\listT}} = x$.
Assume by contradiction that $x \in \listCons{h}{\listT}$.
Then we must have $x \in \listT$.
By hypothesis, every element of $\listCons{h}{\listT}$ is greater than or equal to $x$, hence $h \geq x$.
Moreover, since $\listCons{h}{\listT}$ is ordered and duplicate-free, and since $x \in \listT$, we have that $h < x$ --- contradiction.
\hfill\qedex
\end{itemize}

\noindent
We now apply \ref{eq:find_gap_notin}
on $\listL = \gaps{M}$ and $x = 1$.
Indeed, $\gaps{M} \in \ltlists$ and its elements are greater than or equal to $1$, hence $\ml{1}{\gaps{M}} = \length{\gaps{M}}$.
Therefore, $\findgap[1]{\gaps{M}} \notin \gaps{M}$, which means that $\findgap[1]{\gaps{M}} \in M$.

\medskip
To conclude, we have to show that $\findgap[1]{\gaps{M}}$ is the minimum non-zero element of $M$, \ie for all $n \in M \setminus \setenum{0}$, it must be
$\findgap[1]{\gaps{M}} \leq n$.
Again, we obtain this as a special case of a more general lemma: 
\begin{equation}
    \label{eq:find_gap_le_lt_in}
    \forall x,n \in \N, \;\; 
    \forall \listL \in \ltlists, \;\;
    x \le n < \findgap[x]{\listL}
    \; \rightarrow \; 
    n \in \listL
    \tag*{\citeCoq[find_gap_le_lt_in]{list_nat.v\#L47}}
\end{equation}

\noindent
We prove this by induction on the structure of $\listL$.
For the base case $\listL = \listEmpty$, we have that $\findgap[x]{\listL} = x$ and there is no such $n$.
For the inductive case, let $\listL = \listCons{h}{\listT}$
and let $n$ such that $x \le n < \findgap[x]{\listCons{h}{\listT}}$.
If $h = n$, the thesis $n \in \listCons{h}{\listT}$ holds trivially.
Otherwise, if $h \neq n$, there are two subcases.
If $h \neq x$ then $\findgap[x]{\listCons{h}{\listT}} = x$ and the thesis holds trivially since the hypothesis $x \le n < x$ is false.
Otherwise, if $h = x$, then $\findgap[x]{\listCons{h}{\listT}} = \findgap[{x+1}]{\listT}$.
Since $x \neq n$, then $x+1 \le n < \findgap[{x+1}]{\listT}$, and so $n \in t$ follows by the inductive hypothesis.
\hfill\qedex

\medskip\noindent
We now apply \ref{eq:find_gap_le_lt_in}
with $\listL = \gaps{M}$ and $x=1$.
The hypothesis $n \in M \setminus \setenum{0}$ means that $n \not\in \gaps{M}$, hence by the contrapositive we obtain that 
$1 \not \leq n$ or $n \not< \findgap[1]{\gaps{M}}$.
The first disjunct is impossible since $n \neq 0$, hence it must be
$n \geq \findgap[1]{\gaps{M}}$ --- which gives the thesis.
\qed
\end{proof}

\subsection{Computing the Apéry set}\label{sec: Apery}

We present two algorithms for computing the Apéry set of a numerical semigroup $M$ given the list of its gaps $\gaps{M}$. Moreover, we provide a formalization of a key  theoretical result, \ie that the Apéry set of $M$ is a generator of $M$ (\Cref{th:apery_generates}).

Recall that the Apéry set of a numerical semigroup $M$ \wrt $n$ is defined as:
\[
    \Set{x \in M | x - n \notin M},
\]
where $x - n$ is the difference in $\Z$.
In particular, if $x < n$ then $x - n$ is negative and so it does not belong to $M$.
Since we prefer our Rocq formalization to work with natural numbers only, we instead use truncated subtraction, \ie $x - n = 0$ when $x < n$, and thus we would have $x - n = 0 \in M$.
In order to reflect this difference, we amend the definition of Apéry set as follows.
\begin{equation}
    \label{eq:apery-nat}
    \Set{x \in M | n \le x \; \rightarrow \; x - n \notin M}.
\end{equation}

As for the other invariants, we represent the Apéry set as a list of natural numbers. 
\Cref{def:apery} gives an algorithm to compute this list from the list of gaps $\gaps{M}$ representing $M$.
\Cref{th:apery_spec} establishes the correctness of this algorithm.

\begin{definition}
    \label{def:apery}
    \citeCoq[apery]{apery.v\#L18}
    Let $\listG \in \natlists$ and $n \in \N$. Let $\listL = [0;\ldots;n-1]$, and let $\listT$ be the list obtained by adding $n$ to every element of $\listG$.
    We define $\apery{\listG}{n}$ as the list obtained by removing the elements in $\listG$ from the list $\app{\listL}{\listT}$.
\end{definition}

\begin{theorem}
    \label{th:apery_spec}
    \citeCoq[apery_spec]{apery.v\#L42}
    Let $M$ be a numerical semigroup with gaps list $\gaps{M}$ and let $n \in \N$.
    Then, $\apery{\gaps{M}}{n}$ coincides with the Apéry set of $M$ with respect to $n$.
\end{theorem}
\begin{proof}
    Let $x \in \apery{\gaps{M}}{n}$.
    By~\Cref{def:apery}, we have $x \notin \gaps{M}$ and $x \in \app{\listL}{\listT}$, where \mbox{$\listL = [0;\dots; n-1]$} and $\listT$ is obtained by adding $n$ to every element of $\gaps{M}$.
    %
    If $x \in \listL$ then $x < n$, and so the implication in~\eqref{eq:apery-nat} holds trivially.
    If $x \in \listT$ then $x = y + n$ with $y \in \gaps{M}$.
    Hence, $x - n = y \not\in M$ since $\gaps{M}$ is the list of gaps,
    and again the condition in~\eqref{eq:apery-nat} is true.
    Suppose now that $x \in \Set{x \in M | n \le x \rightarrow x - n \notin M}$.
    If $n \le x$ then $x-n \notin M$, and thus $x-n \in \gaps{M}$.
    We can write $x$ as $x = (x-n) + n$, and thus $x \in \listT$.
    Otherwise, if $x < n$, then clearly $x \in \listL$.
    \qed
\end{proof}

\begin{example}
    \label{ex:apery}
    Let $\gaps{M} = [1;2;3;5;6;9;13]$ be the gaps list representing the numerical semigroup $M$ generated by $\setenum{4;7;10}$.
    Following \Cref{def:apery}, for $n = 4$ we have $\listL = [0;1;2;3]$ and $\listT = [5;6;7;9;10;13;17]$.
    By removing the elements of $\gaps{M}$ from the list $\app{\listL}{\listT}$, we obtain $\apery{\gaps{M}}{n} = [0;7;10;17]$.
    \hfill\qedex
\end{example}

We present in~\Cref{th:apery_2} an alternative algorithm for computing the Apéry set in the special case $n \in M \setminus \setenum{0}$.
\Cref{th:apery_2_correct} proves the correctness of this algorithm.

\begin{definition}
    \label{th:apery_2}
    \citeCoq[apery_2]{apery.v\#L340}
    Let $\listL \in \natlists$ and let $n, a \in \N$.
    We define the function $\findmod{\listL}{n}{a}$ as follows:
    \begin{itemize}
        \item $\findmod{\listL}{n}{a}$ is the first element $x$ in $\listL$ such that $x \equiv_n a$, if such element exists;
        \item otherwise, $\findmod{\listL}{n}{a}$ is the smallest natural number $x$ that is greater than the last element of $\listL$ and such that $x \equiv_n a$.
    \end{itemize}
    We define the list $\aperytwo{\listL}{n}$ as $[\findmod{\listL}{n}{0}; \dots; \findmod{\listL}{n}{n-1}]$.
\end{definition}

\begin{theorem}
    \label{th:apery_2_correct}
    \citeCoq[apery_2_correct]{apery.v\#L342}
    Let $M$ be a numerical semigroup and let $n \in M\setminus\{0\}$.
    Let $\listL$ be the list of small elements of $M$.
    Then $\aperytwo{\listL}{n}$ coincides with the Apéry set of $M$ with respect to $n$.
\end{theorem}
Clearly $\findmod{M}{n}{a}$ is the smallest element $x$ of $M$ such that $x \equiv_n a$.
Thus, \Cref{th:apery_2_correct} follows by the next theoretical \namecref{th:apery_spec_2}.

\begin{lemma}
    \citeCoq[apery_spec_2]{apery.v\#L122}
    \label{th:apery_spec_2}
    Let $M$ be a numerical semigroup and let $n \in M \setminus \setenum{0}$.
    For all $i = 0, \dots, n-1$ let $w(i)$ be the smallest element of $M$ such that $w(i) \equiv_n i$.
    Then the Apéry set of $M$ with respect to $n$ is equal to the set:
    \[
        \Set{w(0), \dots, w(n-1)}
    \]
\end{lemma}
We do not present the proof of this \namecref{th:apery_spec_2} here, since it is essentially a formalization of the (informal) proof given in \cite{numsgps}.

\begin{example}
    Using this second algorithm, we compute the Apéry set relative to $4$ of the same numerical semigroup $M$ presented in \Cref{ex:apery}.
    Given the list of gaps $\gaps{M} = [1;2;3;5;6;9;13]$, the list of small elements of $M$ is $\listL = [0;4;7;8;10;11;12;14]$.
    The smallest element of $\listL$ congruent to $0$ modulo $4$ is $\findmod{\listL}{4}{0} = 0$.
    Similarly, we have $\findmod{\listL}{4}{2} = 10$ and $\findmod{\listL}{4}{3} = 7$.
    No element of $\listL$ is congruent to $1$ modulo $4$, and the smallest natural number after $14$ that satisfies such property is $\findmod{\listL}{4}{1} = 17$.
    Thus, $\aperytwo{\listL}{4} = [0;17;10;7]$.
\end{example}

\begin{theorem}
    \label{th:apery_generates}
    \citeCoq[apery_generates]{apery.v\#L220}
    Let $M$ be a numerical semigroup with gaps list $\gaps{M}$ and let $n \in M\setminus\{0\}$.
    For all $a \in M$ there exists a unique $(k, w) \in \N \times \apery{\gaps{M}}{n}$ such that $a = kn+w$.
\end{theorem}

It follows from \Cref{th:apery_generates} that $\apery{\gaps{M}}{n} \cup \{n\}$ is a (finite) set of generators for $M$.

\subsection{Equivalent definitions}

We recall in~\Cref{th:numerical_semigroup_2} some equivalent definitions of numerical semigroups (see \cite[Proposition 1]{numsgps} and its proof). For the reader's convenience we include (part of) the proof adapted to our Rocq formalization of numerical semigroups.

\begin{theorem}
\label{th:numerical_semigroup_2}
\citeCoq[numerical_semigroup_2]{def.v\#L109}
Let $M$ be an additive submonoid of $\N$ with a decidable membership predicate.
Then, the following statements are equivalent: 
\begin{enumerate}
    \item $M$ is a numerical semigroup (\ie, $\mathbb{N}\setminus M$ is finite); 
    \item  $1\in G$, where $G$ is the additive subgroup of $\mathbb{Z}$ generated by $M$; 
    \item there exists an element $a \in M$ such that $a+1 \in M$.
\end{enumerate}
\end{theorem}
\begin{proof}
$(2 \Leftrightarrow 3)$ easily follows from the fact that $G = \setcomp{x-y}{x, y \in M}$.

\noindent
$(1 \Rightarrow 3)$. If $M$ is a numerical semigroup with conductor $c$, then $c, c+1 \in M$.

\noindent
$(3 \Rightarrow 1)$.
It is easy to prove (see \cite[Proposition 1]{numsgps}) that for every $n \in \N$, if $n \ge (a-1) (a+1)$ then $n \in M$.
This can be used to build a gaps list for $M$.
Indeed, let $\listL$ be the list of natural numbers from $0$ to $(a-1) (a+1)-1$, and let $\listG$ be the list obtained by removing from $\listL$ the elements that are in $M$.
Then, $\listG$ is sorted in ascending order, it has no duplicates, and contains exactly the elements that are not in $M$.
Thus, $M$ together with $\listG$ is a numerical semigroup.
\qed
\end{proof}


\paragraph{Remark.}
Note that the proof of $(3 \Rightarrow 1)$ is constructive: it not only provides a proof that a numerical semigroup on $M$ exists, but also outputs that numerical semigroup (by computing its list of gaps $\listG$).
The decidability of the membership of $M$ is needed in order to construct the list $\listG$.
More precisely, in order to filter out from $\listL$ the elements not belonging to $M$, we need an algorithm that computes whether a natural number belongs to $M$ or not.
Observe also that condition 3 in the above proposition is much easier to implement in Rocq than condition 2, as it does not require introducing $\Z$ nor any of its (additive) subgroups. 

%% file: generators.tex
\section{Computing a numerical semigroup from its generators}
\label{sec:gen-to-small-elements}

In this section we present an algorithm that takes as input a list $\gen$, and gives as output the numerical semigroup $M$ generated by $\gen$, represented by its gaps list $\gaps{M}$.
Our algorithm works as follows: 
\begin{enumerate}
    \item we start by giving an algorithm that generates \emph{multiplicity index lists}, which encode all (and only) the linear combinations of elements in $\gen$ (\Cref{sec:gen});
    \item based on this algorithm, we show how to obtain the actual linear combinations of elements of $\gen$
    (\Cref{sec:gen-to-linear-combinations});
    \item we compute lists $\listL$ of increasing length containing the $\length{\listL}$ smallest elements of $M$ (\Cref{sec:linear-combinations-to-small-elements});
    \item we stop when a list $\listL$ contains $m$ consecutive natural numbers, where $m$ is the minimum non-zero element of $\gen$;
    \item the list of small elements is the list formed by the element of $\listL$ up to $c$, where $c$ is the first of the $m$ consecutive elements of $\listL$ found in the previous step;
    \item finally, we compute the list $\gaps{M}$ of gaps by removing the small elements from the list $[0; \dots; c]$.
\end{enumerate}



\subsection{Generating multiplicity index lists}
\label{sec:gen}


Let $\gen$ be a list of generators. 
Any linear combination of $\gen$ can be represented as a multiset of elements in $\gen$ (thus, not mentioning the coefficients $\lambda_i \in \mathbb{N}$).
For example, given $\gen = [4;7;10]$, we can represent \mbox{$45 = 2\cdot 4 + 1\cdot 7 + 3\cdot 10$} as the multiset $\setenum{4, 4, 7, 10, 10, 10}$, namely $45$ is the sum of its elements.
This multiset, in turn, can be represented as a list $\listL$ whose elements are \emph{indices} of elements in~$\gen$: each index $i$ is repeated as many times as the number of occurrences of $\gen[i]$ in the multiset.
In our example, the index $0$ (corresponding to the generator~$4$) is repeated 2 times, and so on.
In order to have a \emph{unique} representation for each linear combination, we require the list $\listL$ to be sorted in decreasing order.
Therefore, in the previous example, we encode $45$ as $\listL = [2;2;2;1;0;0]$.

Thus, generating linear combinations of $\gen$ is equivalent to generating decreasing lists of natural numbers less than $\length{\gen}$,
\ie lists in the set $\gelists{m}$ where $m = |\gen|-1$.
Such generator is given by the algorithm $\nh{m}{}$ (\Cref{def:lgen}),
which takes as input a natural number $n$, and gives as output the $n$-th multiplicity index list.
This is done recursively, by calling the function $\next[m]{}$ (\Cref{def:next}) on the list obtained at step $n-1$.
The main results in this section are \Cref{th:lgen_complete,th:lgen_correct}, which show that $\nh{m}{}$ generates all, and only, the lists in~$\gelists{m}$. 

\begin{definition}
    \label{def:next}
    \citeCoq[next]{list_alg.v\#L23}
    For all $m \in \N$,
    we define the function $\next[m]{} \colon \natlists \rightarrow \natlists$ 
    by the following inference rules:
    \[
    \begin{array}{c}
	\next[m]{\listEmpty} = [0]
        \qquad\qquad
	\next[m]{\listCons{h}{\listT}} = \listCons{(h+1)}{\listT}
        \;\; \text{if $h < m$}
        \\[10pt]
        \irule
            {\next[m]{\listT} = \listEmpty}
            {\next[m]{\listCons{h}{\listT}} = \listEmpty}
        \;\;\text{if $h \ge m$}
	\qquad
        \irule
            {\next[m]{\listT} = \listCons{x}{\listTi}}
            {\next[m]{\listCons{h}{\listT}} = 
	\listCons{x}{\listCons{x}{\listTi}}}
        \;\;\text{if $h \ge m$}
    \end{array}
    \]
\end{definition}

\noindent
\Eg, we have the following derivation, proving that
$\next[3]{[3;3;1;0]} = [2;2;2;0]$:
\[
\irule
{\irule
  {\irule
    {}
    {\next[3]{1;0} = [2;0]}}
  {\next[3]{[3;1;0]} = [2;2;0]}}
{\next[3]{[3;3;1;0]} = [2;2;2;0]}
\]
Note that $\next[m]{\listT}$ cannot be the empty list, and so the third rule is never used.
However, it is still needed to provide a correct definition of $\next{}$ in Rocq.

Intuitively, $\next[m]{\listL}$ searches the leftmost value in $\listL$ that is (strictly) less than $m$, increases that value by one, and changes all the elements at its left to that same (increased) value.
If all the elements in the list are greater than or equal to $m$, it adds a new $0$ element at the end of the list, and sets all the other elements to $0$.
For example, given $m = 3$, we have the following mappings:
\[
    [3;3;1] 
    \xmapsto{\next[3]{}}     
    [2;2;2]
    \xmapsto{\next[3]{}} 
    [3;2;2] 
    \xmapsto{\next[3]{}} 
    [3;3;2] 
    \xmapsto{\next[3]{}} 
    [3;3;3] 
    \xmapsto{\next[3]{}} 
    [0;0;0;0] 
\]


\noindent
This intuition is formalized by the following~\namecref{lem:next_repeat}.  

\begin{lemma}
    \label{def:next_repeat}
    \label{lem:next_repeat}
    \citeCoq[next_repeat]{list_alg.v\#L40}
    For all $\listL \in \natlists$ and all $m, n, k$ in $\N$ with $k < m$:
    \[
	\next[m]{\app{\rep{m}{n}}{(\listCons{k}{\listL})}}
        = 
        \app{\rep{k+1}{n+1}}{\listL}
    \]
\end{lemma}
\begin{proof}
    By induction on $n$.
    For the base case $n = 0$, since $k < m$, by the second inference rule in~\Cref{def:next} we have that:
    \[
    \next[m]{\rep{m}{0} \listConcat (\listCons{k}{\listL})} 
    \; = \; 
    \next[m]{\listCons{k}{\listL}} 
    \; = \;
    \listCons{(k+1)}{\listL}
    \; = \;
    \rep{k+1}{1} \listConcat \listL
    \]
    For the inductive case, let the induction hypothesis be:
    \begin{equation}
    \label{eq:next_repeat:ih}
    \next[m]{\app{\rep{m}{n}}{(\listCons{k}{\listL})}}
    \; = \; 
    \app{\rep{k+1}{n+1}}{\listL}
    \; = \;
    \listCons{(k+1)}{\app{\rep{k+1}{n}}{\listL}}
    \end{equation}
    Then, by the fourth inference rule in~\Cref{def:next} we have that:
    \begin{align*}
    \next[m]{\app{\rep{m}{n+1}}{(\listCons{k}{\listL})}}
    & = 
    \next[m]{\listCons{m}{(\app{\rep{m}{n}}{(\listCons{k}{\listL})}})}
    && \text{since $\rep{m}{n+1} = \listCons{m}{\rep{m}{n}}$}
    \\
    & = \listCons{(k+1)}{\listCons{(k+1)}{\app{\rep{k+1}{n}}{\listL}}}
    && \text{by~\eqref{eq:next_repeat:ih}}
    \\
    & =
    \app{\rep{k+1}{n+2}}{\listL}
    && \tag*{\qed}
    \end{align*}
\end{proof}

We can now proceed to define our list generating algorithm as the iteration of the function $\next[m]{}$, starting from the empty list.

\begin{definition}
    \label{def:lgen}
    \citeCoq[lgen]{list_alg.v\#L207}
    For all $m \in \N$,
    we define $\nh{m}{} \colon \N \rightarrow \natlists$ as follows:
    \[
    \nh{m}{0} = \listEmpty 
    \qquad
    \nh{m}{x+1} = \next[m]{\nh{m}{x}}
    \]
\end{definition}

\noindent
For illustration, we tabulate below the first 12 values of $\nh{2}{}$:
\label{ex:genL}
\begin{align*}
    \nh{2}{0} &= \listEmpty & \nh{2}{1} &= [0] & \nh{2}{2} &= [1] & \nh{2}{3} &= [2] \\
    \nh{2}{4} &= [0;0] & \nh{2}{5} &= [1;0] & \nh{2}{6} &= [2;0] & \nh{2}{7} &= [1;1] \\
    \nh{2}{8} &= [2;1] & \nh{2}{9} &= [2;2] & \nh{2}{10} &= [0;0;0] & \nh{2}{11} &= [1;0;0]
\end{align*}

We now prove that $\nh{m}{}$ generates exactly the lists in $\gelists{m}$.

\begin{theorem}
    \label{th:lgen_correct}
    \citeCoq[lgen_correct]{list_alg.v\#L230}
    For all $m \in \N$, $\ran{\nh{m}{}} \subseteq \gelists{m}$.
\end{theorem}
\begin{proof}
    This is a direct consequence of the auxiliary result:
    \begin{equation}
        \forall \listL \in \gelists{m}, \;\; \next[m]{\listL} \in \gelists{m}
        \tag*{\citeCoq[gelist_next]{list_alg.v\#L82}}
    \end{equation}
    We prove it by induction on $\listL$.
    For the base case $\listL = \listEmpty$, we have $\next[m]{\listL} = [0] \in \gelists{m}$.
    For the inductive case, let $\listL = \listCons{h}{\listT}$ be such that $\listL \in \gelists{m}$.
    Then also $\listT \in \gelists{m}$ and thus, by the inductive hypothesis, $\next[m]{\listT} \in \gelists{m}$.
    
    We have two cases.
    If $h < m$, then $\next{\listCons{h}{\listT}} = \listCons{(h+1)}{\listT}$, which is in $\gelists{m}$ since $\listCons{h}{\listT} \in \gelists{m}$ and $h+1 \leq m$.
    If $h = m$, then $\next{\listCons{h}{\listT}} = \listCons{x}{\listCons{x}{\listTi}}$, where $\listCons{x}{\listTi} = \next{\listT}$.
    Since $\next{\listT} \in \gelists{m}$, we conclude that $\next{\listCons{h}{\listT}} \in \gelists{m}$.
    \qed
\end{proof}

Proving the other inclusion, stated by~\Cref{th:lgen_complete}, is more complex.
Indeed, the idea is to prove this \namecref{th:lgen_complete} by induction, but a simple induction on the list structure does not work, since we cannot easily derive that a list $\listCons{h}{\listT}$ belongs to $\ran{\nh{m}{}}$ if its tail $\listT$ does.
We first state the~\namecref{th:lgen_complete}, and then present the  intuitions, definitions and results that are needed for its proof. 

\begin{theorem}
    \label{th:lgen_complete}
    \citeCoq[lgen_complete]{list_alg.v\#L659}
    For all $m \in \N$,
    $\gelists{m} \subseteq \ran{\nh{m}{}}$.
\end{theorem}

To prove~\Cref{th:lgen_complete}, we first define an operator that gives a lexicographic ordering between the lists generated by $\nh{m}{}$.
This ordering will then be exploited in our inductive proof.
Formally, for all $n \in \N$, we define the operator $\mls{n}{\,\cdot\,} \colon \natlists \to \N^{n+1}$ as follows:
\[
    \mls{n}{\listL} 
    \quad \eqdef \quad
    \bigl(\ml{0}{\listL}, \dots, \ml{n}{\listL}\bigr)
\]
Note that the order in which $\nh{m}{}$ generates lists  corresponds to the lexicographic order associated with $\mls{m}{\listL}$, as exemplified by the following table:
\[
\begin{array}{r@{\qquad}*{4}{c}}
    \listL	& \ml{0}{\listL} & \ml{1}{\listL} & \ml{2}{\listL} & \ml{3}{\listL} \\
    \midrule
    \nh{3}{x} = [3;3;2;0]	& 4    & 3    & 3    & 2    \\
    \nh{3}{x+1} = [3;3;3;0]	& 4    & 3    & 3    & 3    \\
    \nh{3}{x+2} = [1;1;1;1]	& 4    & 4    & 0    & 0    \\
    \nh{3}{x+3} = [2;1;1;1]	& 4    & 4    & 1    & 0    \\
\end{array}
\]

Following this intuition, we prove that if all the lists $\listL$ with given values for $\ml{0}{\listL}, \dots, \ml{n}{\listL}$ are generated, then all the lists with the same given values, except for $\ml{n}{\listL}$, which is increased by one, are generated.
This is expressed by the following definition and by~\Cref{lem:le_length_eq_impl}.

\begin{definition}
    \label{lem:le_length_eq_lgen}
    \citeCoq[le_length_eq_lgen]{list_alg.v\#L405}
    For all $m, n \in \N$ and $k \in \N^{n+1}$ we define the proposition $\all{m}{n}{k}$ as follows:
    \[
        \all{m}{n}{k}
        \quad \eqdef \quad
        \forall \listL \in \gelists{m},\;\;
        \mls{n}{\listL} = k
        \rightarrow
        \listL \in \ran{\nh{m}{}}
    \]
    Moreover, we define $\fsucc{k} = (k_0, \dots, k_{n-1}, k_n+1)$.
\end{definition}

\begin{lemma}
    \label{lem:le_length_eq_impl}
    \citeCoq[le_length_eq_impl]{list_alg.v\#L601}
    For all $m, n \in \N$ and $k \in \N^{m-n+1}$, if $\all{m}{m-n}{k}$ holds then $\all{m}{m-n}{\fsucc{k}}$ holds.
\end{lemma}

As a preliminary, we prove a special case.

\begin{lemma}
    \label{lem:le_length_0_lgen}
    \citeCoq[le_length_0_lgen]{list_alg.v\#L437}
    Let $m, n \in N$ and $k \in \N^{n+1}$ be such that $n \le m$ and $\all{m}{n}{k}$ holds.
    If $\mls{n}{\listL} = \fsucc{k}$ and $\ml{n+1}{\listL} = 0$ then $\listL \in \ran{\nh{m}{}}$.
\end{lemma}
\begin{proof}
    Since $\listL$ is sorted and $\ml{n+1}{\listL} = 0$, all the elements of $\listL$ up to the position $\ml{n}{\listL} = k_n+1$, and not more, are equal to $n$; that is, $\listL = \app{\rep{n}{k_n+1}}{\listT}$, where the elements of $\listT$ are less than $n$.
    There are two cases:
    \begin{itemize}    
    \item If $n \neq 0$, let $\listLi = \app{\rep{m}{k_n}}{\listCons{(n-1)}{\listT}}$.
    The list $\listLi$ satisfies the hypothesis of $\all{m}{n}{k}$, and thus it belongs to $\ran{\nh{m}{}}$.
    By \Cref{lem:next_repeat}, $\next[m]{\listLi} = \listL$.
    \item If $n = 0$ then $\listL = \rep{0}{k_0+1}$, which is the successor of the list $\rep{m}{k_0}$.
    This list belongs to $\ran{\nh{m}{}}$ because it satisfies the hypothesis of $\all{m}{0}{k}$.
    \qed
    \end{itemize}
\end{proof}


\begin{proof}[\Cref{lem:le_length_eq_impl}]
    By induction on $n$.
    For the base case $n = 0$, let $\listL$ satisfy the hypothesis of $\all{m}{m}{\fsucc{k}}$.
    In particular $\ml{m+1}{\listL} = 0$, since the elements of $\listL$ are all less than or equal to $m$.
    We conclude by \Cref{lem:le_length_0_lgen}.

    For the inductive case, let $a = m - n - 1$.
    The inductive hypothesis is:
    \begin{equation}
        \label{eq:all_impl}
        \forall k \in \N^{a+2},\;\; \all{m}{a+1}{k} \rightarrow \all{m}{a+1}{\fsucc{k}}
    \end{equation}
    Let $k \in \N^{a+1}$ and let $\listL$ be a list that satisfies the hypothesis of $\all{m}{a}{\fsucc{k}}$.
    Let $x = \ml{a+1}{\listL}$.
    We prove by induction on $x$ that $\listL \in \ran{\nh{m}{}}$.
    If $x = 0$ we conclude by \Cref{lem:le_length_0_lgen}.

    The inductive hypothesis is the following: for every $k \in \N^{a+1}$ and for every list $\listL$ that satisfies the hypothesis of $\all{m}{a}{\fsucc{k}}$ and such that $\ml{a+1}{\listL} = x$, we have that $\listL \in \ran{\nh{m}{}}$.
    Let $k \in \N^{a+1}$ and let $\listL$ be a list that satisfies the hypothesis of $\all{m}{a}{\fsucc{k}}$ and such that $\ml{a+1}{\listL} = x+1$.
    Let $k' = (k_0, \dots, k_a, x)$.
    By inductive hypothesis on $x$, the proposition $\all{m}{a+1}{k'}$ holds.
    Thus, by \ref{eq:all_impl} we deduce that $\all{m}{a+1}{\fsucc{k'}}$ holds.
    The list $\listL$ satisfies the hypothesis of $\all{m}{a+1}{\fsucc{k'}}$, hence $\listL \in \ran{\nh{m}{}}$.
    \qed
\end{proof}

\begin{proof}[\Cref{th:lgen_complete}]
    Let $\listL \in \gelists{m}$.
    We prove that $\listL \in \nh{m}{}$ by induction on the length of $\listL$.
    For the base case, if $\length{\listL} = 0$ then $\listL = \listEmpty = \nh{m}{} 0$.
    Suppose that every list of length $k$ belongs to $\ran{\nh{m}{}}$.
    Then, equivalently, the proposition $\all{m}{0}{k}$ holds.
    Thus, by \Cref{lem:le_length_eq_impl}, the proposition $\all{m}{0}{\fsucc{k}} = \all{m}{0}{k+1}$ holds, which is equivalent to: every list of length $k+1$ belongs to $\ran{\nh{m}{}}$.
    \qed
\end{proof}

\subsection{From multiplicity index lists to linear combinations} 
\label{sec:gen-to-linear-combinations}

We now apply our $\nh{m}{}$ algorithm to obtain all the linear combinations of the (list of) generators $\gen$ of a numerical semigroup.
We do this by introducing in~\Cref{def:mgen} a function $\nhs{\gen}{} \colon \N \rightarrow \N$ that generates the linear combinations of $\gen$.
Namely, $\nhs{\gen}{n}$ is the linear combination of elements of $\gen$ corresponding to the list $\nh{m}{n}$, where $m$ is the bound on the indices of the multiplicity index lists.
We will prove in~\Cref{th:mgen_complete} that $\nhs{\gen}{}$ iteratively generates all the  linear combinations of elements of $\gen$, and thus all possible elements of the submonoid of $\N$ generated by $\gen$.

\begin{definition}
    \label{def:mgen}
    \citeCoq[mgen]{generators.v\#L109}    
    For all $\gen \in \natlists$,
    we define $\nhs{\gen}{} \colon \N \rightarrow \N$ as:
    \[
    \nhs{\gen}{n} 
    \; =
    \sum_{i \in \nh{\length{\gen}-1}{n}} \!\!\! \gen[i]
    \]
\end{definition}

\noindent
\Eg, recalling the tabulation of $\nh{2}{}$ at page~\pageref{ex:genL},
for $\gen = [4;7;10]$, we have:
\begin{align*}
\nhs{\gen}{3} & = \sum_{i \in \nh{2}{3}} \!\!\! \gen[i] = \sum_{i \in [2]}  \gen[i] = \gen[2] = 10
\\
\nhs{\gen}{4} & = \sum_{i \in \nh{2}{4}} \!\!\! \gen[i] = \sum_{i \in [0;0]}  \gen[i] = \gen[0] + \gen[0] = 4+4 = 8
\end{align*}
which shows that $\nhs{\gen}{}$ is \emph{not} monotonic, \ie
$n \leq n' \not\rightarrow \nhs{\gen}{n} \leq \nhs{\gen}{n'}$.

\begin{theorem}
    \label{th:mgen_complete}
    \citeCoq[mgen_complete]{generators.v\#L144}
    Let $M$ be the submonoid of $\N$ generated by $\gen \in \natlists$, where $0 \not\in \gen \neq \listEmpty$.
    Then $M \subseteq \ran{\nhs{\gen}{}}$.
\end{theorem}
\begin{proof}
    Let $d \in M$.
    Since $M$ is generated by $\gen$, its element $d$ can be written as a linear combination of $\gen$, \ie, there exist $\lambda_1, \ldots, \lambda_n \in \N$ 
    and $x_1,\ldots,x_n \in \gen$ such that $d = \sum_1^n \lambda_i x_i$.
    For $i \in \setenum{1,\ldots,n}$, let $p_i$ be the index of $x_i$ in $\gen$, starting from $0$.
    Let $\listL$ be the list obtained by sorting the list
    $[p_1]^{\lambda_1} \listConcat \dots [p_n]^{\lambda_n}$
    in descending order. 
    Let $m = \length{\gen}-1$.
    Since $\listL \in \gelists{m}$, by~\Cref{th:lgen_complete} there exists $k \in \N$ such that $\nh{m}{k} = \listL$.
    From the definition of $\listL$ it follows that
    \[
    \nhs{\gen}{k} 
    \; = \; 
    \sum_{i \in \listL} \gen[i] 
    \; = \; 
    d 
    \tag*{\qed}
    \]
\end{proof}


Since $\nhs{\gen}{}$ does not generate the elements of $M$ in ascending order, we cannot just obtain the small elements of $M$ by applying $\nhs{\gen}{}$ until finding the conductor.
%
%
Still, our algorithm $\nhs{\gen}{}$ satisfies the following property: there exist certain numbers $n$ (and they are infinite) such that, when the algorithm has generated $n$, we know that all possible linear combinations $d \leq n$ have already been generated.
Thus, we just need to run the algorithm until one of those numbers that is greater than or equal to the conductor has been generated.

\subsection{From linear combinations to small elements} 
\label{sec:linear-combinations-to-small-elements}

We present our algorithm to compute the list of small elements of a numerical semigroup in \Cref{algo:small-elements}.
For simplicity, there we rely on pseudocode; see \citeCoq[small_els]{generators.v\#L392} for our actual Rocq implementation.

\begin{algorithm}
\footnotesize
\caption{Algorithm for computing the small elements given the generators}
\label{algo:small-elements}
\begin{algorithmic}
    \Function{Consecutive-Values}{$\listL,m, l$}
    \Comment{Search a subsequence of length $m$}
        \State $\mathit{next} \gets 0; \quad \mathit{count} \gets 0; \quad i \gets 0;$
        \While{$i < \length{\listL}$}
            \State \textbf{if} $\listL[i] = \mathit{next}$
                \textbf{then} $\mathit{count} \gets \mathit{count}+1;$
            \State \textbf{else}
            (\textbf{if} $\listL[i] \le l$ \textbf{then} $\mathit{count} \gets 1$ \textbf{else} \Return $-1;$)
            \State \textbf{if} $\mathit{count} = m$ \textbf{then} \Return $i+1-m;$
            \Comment{Index of 1st element of subsequence}
            \State $\mathit{next} \gets \listL[i]+1; \quad i \gets i+1;$
        \EndWhile
        \State \Return $-1;$
    \EndFunction

    \Statex

    \Function{Small-Elements}{$\gen,n$}
        \State $\vec{se} = \listEmpty; \quad \vec{l} \gets \listEmpty; \quad m \gets \min \gen; \quad i \gets 1;$
        \While{$i < n$}
            \While{$\length{\vec{l}} < i$}
                \State $\vec{se} \gets \listCons{\big(\sum_{j \in \vec{l}} \gen[j]\big)}{\vec{se}};$
                \State $\vec{l} \gets \next[\length{\gen}-1]{\vec{l}};$
            \EndWhile
            \State $\vec{se} \gets \listCons{(i \cdot m)}{\vec{se}};$
            \State $\vec{l} \gets \next[\length{\gen}-1]{\vec{l}};$
            \State $\vec{se} \gets \textsc{Sort}(\vec{se});$
                \Comment {Sort in ascending order}
            \State $\vec{se} \gets \textsc{RmDups}(\vec{se});$
                \Comment{Remove duplicates}
            \State $p \gets \textsc{Consecutive-Values}(\vec{se},m,i \cdot m);$
            \State \textbf{if} $p \geq 0$ \textbf{then} \Return $\textsc{Take}(\vec{se}, p+1);$
                \Comment {Prefix of $\vec{se}$ of length $p+1$}
            \State $i \gets i+1;$
        \EndWhile
    \EndFunction
\end{algorithmic}
\end{algorithm}

The function \textsc{Consecutive-Values}$(\listL,m,l)$ \citeCoq[find_seq]{list_nat.v\#L414} returns the index of the first element $x \leq l$ in $\listL$ such that there is a sequence of $m$ consecutive natural numbers starting from $x$ in $\listL$.
The return value $-1$ means that no such sequence exists. 

Let $\gen \in \natlists$ with $0 \notin \gen \neq \listEmpty$ be the generator of the numerical semigroup $M$.
The function \textsc{Small-Elements} starts with an empty list $\vec{se}$ and with $i = 1$.
In each iteration inside the outer \textbf{while} loop, we add the element $\nhs{\gen}{j}$ to the list $\vec{se}$, for every $j \in \N$ such that $\length{\nh{\length{\gen}-1}{j}} = i-1$ and for the smallest $j \in \N$ such that $\length{\nh{\length{\gen}-1}{j}} = i$. In particular, this last element is equal to $i \cdot m$, where $m = \min \gen$.
Next, we sort $\vec{se}$ in ascending order and we remove any duplicates from it.
Lastly, we apply the function \textsc{Consecutive-Values} to the list $\vec{se}$.
If the error value $-1$ is returned, we proceed with another iteration.
Otherwise, let $p$ be the returned value.
The main function \textsc{Small-Elements} returns the list made of the first $p+1$ elements of $\vec{se}$.

\begin{example}
\label{ex:small-elements}
We illustrate~\Cref{algo:small-elements} through an example with $\gen = [4;7;10]$.
At the 3rd iteration of the outer \textbf{while} loop, $\vec{se} = [0;4;7;8;10;11;14;17;20]$.
Since there are no $4$ consecutive elements in this list, we proceed with another iteration.
At the 4th iteration, $\vec{se} = [0;4;7;8;10;11;12;14;15;16;17;18;20;21;24;27;30]$.
The number $14$ is the first of a sequence of $4$ consecutive numbers in $\vec{se}$, and $14 \le 16 = i \cdot m$.
Thus, the list of small elements is $[0;4;7;8;10;11;12;14]$.
Consequently, the list of gaps is $[1;2;3;5;6;9;13]$.
\hfill\qedex
\end{example}

We now prove the correctness of the algorithm, namely that the prefix of $\vec{se}$ of length $p+1$ is the list of small elements of $M$.
Clearly, $\vec{se} \subseteq M$, since every element added to $\vec{se}$ during the execution of the algorithm is a sum of elements in $\gen$.
Next, every natural number $x$ greater than or equal to $\vec{se}_p$ belongs to $M$.
Indeed, $x$ can be obtained by adding $m$ multiple times to one of $c, \dots, c + m - 1$.
Since $c, \dots, c+m-1$ are in $\vec{se}$ and thus in $M$, and since $m \in M$, also $x \in M$.
Moreover, $c$ is the smallest element of $M$ with this property, since $c-1 \notin M$ (otherwise, \textsc{Consecutive-Values} would have returned the position of $c-1$ instead of the position of $c$).

Finally, we have to prove that every element of $M$ that is smaller than $\vec{se}_p$ is in $\vec{se}$.
This is a consequence of the following theorem.

\begin{theorem}
    \label{th:mgen_complete_lt}
    \citeCoq[mgen_complete_lt]{generators.v\#L182}
    Let $\gen \in \natlists$ be such that $0 \not\in \gen \neq \listEmpty$ and let $M$ be the submonoid of $\N$ generated by $\gen$.
    Let $k \in \N$ and let $x \in M$ be such that $x < \length{\nh{\length{\gen}-1}{k}}\cdot \min \gen$. 
    Then there exists $n < k$ such that $\nhs{\gen}{n} = x$.
\end{theorem}


The \namecref{th:mgen_complete_lt} above states that all the elements of $M$ that are less than $\length{\nhs{\gen}{k}} \cdot \min \gen$ are generated before the $k$-th iteration of $\nhs{\gen}{}$.
In other words, after generating $k$ elements of $M$ with $\nhs{\gen}{}$, we are sure that no elements of $M$ smaller than $\length{\nhs{\gen}{k}} \cdot m$ have been forgotten.

For example, recalling the enumeration of $\nh{2}{}$ at page~\pageref{ex:genL}, we have that $\nh{2}{5} = [1;0]$, and thus, for $\gen = [4;7;10]$:
\[
    \length{\nhs{\gen}{5}} \cdot \min \gen
    \; = \; 
    \length{[1;0]} \cdot \min [4;7;10] \; = \; 2 \cdot 4 = 8
\]
Every list generated for $n \ge 5$ has at least length $2$, and thus the corresponding linear combination is greater than or equal to $8$.
This means that every number smaller than $8$ that is in $M$ is necessarily included in the set:
\[
    \setcomp{\nhs{\gen}{i}}{0 \leq i < 5} = \Set{0, 4, 7, 10, 8}
\]

\begin{proof}[\Cref{th:mgen_complete_lt}]
    The proof relies on two auxiliary facts.
    First, we prove that:
    \begin{equation*}
        \label{eq:nth_limit_le_mgen}
        \nhs{\gen}{k} \ge \length{\nh{\length{\gen}-1}{k}} \cdot \min \gen
        \tag*{\citeCoq[nth_limit_le_mgen]{generators.v\#L113}}
    \end{equation*}
    This holds because the term $\nhs{\gen}{k}$ is the sum of $\length{\nh{\length{\gen}-1}{k}}$ elements, each of which is greater than or equal to $\min \gen$.
    \hfill\qedex

    \smallskip\noindent
    Second, we prove that:
    \begin{equation*}
        \label{eq:length_lgen_le}
        \forall m, n, k \in \N, \;
        n \le k
        \; \rightarrow \;
        \length{\nh{m}{n}} \le \length{\nh{m}{k}}
        \tag*{\citeCoq[length_lgen_le]{list_alg.v\#L238}}
    \end{equation*}
    This is a direct consequence of the following auxiliary result:
    \begin{equation*}
    \forall \listL \in \natlists, \; 
    \length{\next[m]{\listL}} \in \setenum{\length{\listL}, \length{\listL}+1}
    \tag*{\citeCoq[length_next]{list_alg.v\#L89}}
    \end{equation*}
    We prove it by induction on $\listL$.
    For the base case, if $\listL = \listEmpty$, 
    then $\length{\next[m]{\listL}} = \length{[0]} = \length{\listL}+1$.
    For the inductive case, let $\listL = \listCons{h}{\listT}$. 
    We have two subcases.
    If $h < m$, then 
    \(
    \length{\next[m]{\listCons{h}{\listT}}} 
    = 
    \length{\listCons{(h+1)}{\listT}}
    =
    \length{\listCons{h}{\listT}}
    \).
    Otherwise, if $h \ge m$, then 
    $\next[m]{\listCons{h}{\listT}} = \listCons{x}{\listCons{x}{\listTi}}$, where $\listCons{x}{\listTi} = \next[m]{\listT}$.
    By the induction hypothesis,
    $\length{\listCons{x}{\listTi}} = \length{\next[m]{\listT}} \in \setenum{\length{\listT}, \length{\listT}+1}$.
    Therefore:
    \begin{align*}
    \length{\next[m]{\listL}} 
    = 
    \length{\listCons{x}{\listCons{x}{\listTi}}}
    =
    1 + \length{\listCons{x}{\listTi}}
    & \in
    1 + \setenum{\length{\listT},\length{\listT}+1} 
    = 
    \setenum{\length{\listL},\length{\listL}+1}
    \tag*{\qedex}
    \end{align*}

    \smallskip\noindent
    We can now prove~\Cref{th:mgen_complete_lt}.
    By~\Cref{th:mgen_complete}, there exists $n \in \N$ such that $\nhs{\gen}{n} = x$.
    Suppose by contradiction that $k \le n$.
    Let $\genlength = \length{\gen}-1$.
    Then, by \ref{eq:length_lgen_le} and \ref{eq:nth_limit_le_mgen} we have
    \[
    x < \length{\nh{\genlength}{k}} \cdot \min \gen
    \le 
    \length{\nh{\genlength}{n}} \cdot \min \gen
    \le \nhs{\gen}{n} = x
    \]
    that is, $x < x$ --- contradiction.
    \qed
\end{proof}

%% file: conclusions.tex
\section{Conclusions}

The present paper provides a first formalization of numerical semigroups in a proof assistant, more precisely Rocq.
In particular, we implemented certified algorithms to compute relevant invariants of a numerical semigroup, \eg its Apéry sets, the small elements, the list of gaps, the multiplicity and the conductor.
Although much more has to be done, this work provides a necessary basis for any further implementation of certified algorithms (and proofs) for numerical semigroups.  
Our formalization is public and counts $\sim 3000$ lines of Rocq code.

While developing our formalization in Rocq, we noted that it is not always convenient to have a perfect match between a mathematical theory and its formalization in a proof assistant.
For instance, in informal mathematics we use $\N$ and $\Z$ interchangeably, while this cannot be easily done in Rocq. For this reason we only work with $\N$, and this choice entails more adaptations, like the definition of Apéry set (\ref{eq:apery-nat}).
Another difference between the informal theory of numerical semigroups and our implementation is that we use lists to represent finite sets.



An algorithm, different from ours, that computes the gaps of a numerical semigroup from a (finite) set of generators is introduced in~\cite{algonumsgps} in the more general framework of generalized numerical semigroups. More specifically, the algorithm computes the Frobenius number $F$ at first, and then, for every $0 \leq n\leq  F$, it checks whether $n$ can be written as linear combination of the generators.
These tasks can be seen as instances of the ``Frobenius problem'', or ``money-changing problem'', which seeks for non-negative integer solutions to $x_1a_1 + \dots + x_n a_n = m$, where $a_i$ and $m$ are positive integers (see \cite{FrobProblem}), for whose solution several algorithms can be applied~\cite{frobenius_problem}. It is worth noticing that since the mentioned procedure requires to check whether each number between $0$ and $F$ is a linear combination (of the generators) or not, it makes sense to tackle the ``reverse'' problem, namely to generate every linear combination instead, which is the first step of our algorithm.

Although the efficiency of our algorithm was not investigated here, its formal proof of correctness is relevant not only intrinsically but also as a benchmark for more efficient algorithms (to be developed in the future), whose correctness can be established by proving their equivalence to our algorithm.
